\documentclass{llncs}
\usepackage{graphicx}
\usepackage{amsmath}
\usepackage{amsfonts}
\usepackage{amssymb}
\usepackage[lined,noend,ruled,linesnumbered,commentsnumbered]{algorithm2e}
\usepackage{microtype}

\usepackage{epsfig}
\usepackage{epstopdf}

\newcommand{\Osymbol}{{\mathcal O}}
\newcommand{\BO}[1]{\Osymbol\left(#1\right)}
\newcommand{\TO}[1]{\tilde{\Osymbol}\left(#1\right)}

\newcommand{\BT}[1]{{\Theta}\left(#1\right)}
\newcommand{\BOM}[1]{\Omega\left(#1\right)}
\newcommand{\sort}[1]{\text{sort}\left({#1}\right)}

\newcommand{\E}[1]{\mathbb{E}\left[#1\right]}

\newcommand{\V}[1]{\text{Var}\left(#1\right)}

\renewcommand{\Pr}[1]{\text{Pr}\left[#1\right]}

\newcommand{\ms}[1]{{#1}}

\newcommand{\SimJoin}{\textsc{OSimJoin}}
\newcommand{\ASimJoin}{\textsc{ASimJoin}}

\begin{document}

\title{I/O-Efficient Similarity Join
\thanks{The research leading to these results has received
funding from the European Research Council under the EU 7th Framework Programme, ERC grant agreement no.~614331.}}
\date{}

\author{Rasmus Pagh \and Ninh Pham \and Francesco Silvestri\thanks{In part supported by University of Padova project CPDA121378 and by MIUR of Italy project AMANDA while working at the University of Padova.} \and Morten St{\"o}ckel\thanks{Supported by the Danish National Research Foundation / Sapere Aude program.}}

\institute{IT University of Copenhagen, Denmark}

\maketitle


\begin{abstract}

We present an I/O-efficient algorithm for computing similarity joins based on locality-sensitive hashing (LSH).
In contrast to the filtering methods commonly suggested our method has provable sub-quadratic dependency on the data size.
Further, in contrast to straightforward implementations of known LSH-based algorithms on external memory, our approach is able to take significant advantage of the available internal memory: Whereas the time complexity of classical algorithms includes a factor of $N^\rho$, where $\rho$ is a parameter of the LSH used, the I/O complexity of our algorithm merely includes a factor $(N/M)^\rho$, where $N$ is the data size and $M$ is the size of internal memory. 
Our algorithm is randomized and outputs the correct result with high probability.
It is a simple, recursive, cache-oblivious procedure, and we believe that it will be useful also in other computational settings such as parallel computation.

\end{abstract}



\section{Introduction}

The ability to handle noisy or imprecise data is becoming increasingly important in computing.
In database settings this kind of capability is often achieved using similarity join primitives that replace equality predicates with a condition on similarity.
To make this more precise consider a space $\mathbb{U}$ and a distance function $d: \mathbb{U}\times \mathbb{U} \rightarrow {\bf R}$.
The {\em similarity join\/} of sets $R, S \subseteq \mathbb{U}$ is the following:
Given a radius $r$, compute the set 
$R \bowtie_{\leq r} S = \{ (x, y)\in R\times S \; | \; d(x, y	)\leq r \}.$
This problem occurs in numerous applications, such as web deduplication~\cite{Bayardo_WWW07,Henzinger_SIGIR06,Xiao_WWW08}, document clustering~\cite{Broder_NETWORK97}, data cleaning~\cite{Arasu_VLDB06,Chaudhuri_ICDE06}. 
As such applications arise in large-scale datasets, the problem of scaling up similarity join for different metric distances is getting more important and more challenging.

Many known similarity join techniques (e.g.,~prefix filtering~\cite{Arasu_VLDB06,Chaudhuri_ICDE06}, positional filtering~\cite{Xiao_WWW08}, inverted index-based filtering~\cite{Bayardo_WWW07}) are based on \emph{filtering} techniques that often, but not always, succeed in reducing computational costs.
If we let $N = |R|+|S|$ these techniques generally require $\Omega(N^2)$ comparisons for worst-case data.
Another  approach is \emph{locality-sensitive hashing} (LSH) where candidate output pairs are generated using collisions of carefully chosen hash functions. The LSH is defined as follows.
\begin{definition}\label{def:LSH}
Fix a distance function $d: \mathbb{U}\times \mathbb{U} \rightarrow {\bf R}$.
For positive reals $r,c,p_1,p_2$, and $p_1 > p_2, c > 1$, a family of functions $\mathcal{H}$ is \emph{$(r,cr,p_1,p_2)$-sensitive} if for uniformly chosen $h\in \mathcal{H}$ and all $x, y\in \mathbb{U}$:
\begin{itemize}
	\item If $d(x,y) \leq r$ then $\Pr{h(x)=h(y)}\geq p_1$;
	\item If $d(x,y) \geq cr$ then $\Pr{h(x)=h(y)}\leq p_2$.
\end{itemize}
We say that $\mathcal{H}$ is \textit{monotonic} if $\Pr{h(x)=h(y)}$ is a non-increasing function of the distance function $d(x,y)$.
We also say that $\mathcal{H}$ uses space $s$ if a function $h\in \mathcal{H}$ can be stored and evaluated using space~$s$.
\end{definition}

LSH is able to break the $N^2$ barrier in cases where for some constant $c > 1$ the number of pairs in $R \bowtie_{\leq cr} S$ is not too large.
In other words, there should not be too many pairs that have distance within a factor $c$ of the threshold, the reason being that such pairs are likely to become candidates, yet considering them does not contribute to the output.
For notational simplicity, we will talk about \emph{far} pairs at distance
greater than $cr$ (those that should not be reported), \emph{near} pairs at
distance at most $r$ (those that should be reported), and \emph{$c$-near} pairs
at distance between $r$ and $cr$ (those that should not be reported but the
LSH provides no collision guarantees).

\smallskip

\textbf{Our contribution.}
In this paper we study I/O-efficient similarity join methods based on LSH.
That is, we are interested in minimizing the number of I/O operations where a block of $B$ points from $\mathbb{U}$ is transferred between an external memory and an internal memory with capacity for $M$ points from $\mathbb{U}$.
Our main result is the first \emph{cache-oblivious} algorithm for similarity join that has \textit{provably} sub-quadratic dependency on the data size $N$ and at the same time inverse polynomial dependency on $M$.
In essence, where previous methods have an overhead factor of either $N/M$ or $(N/B)^\rho$ we obtain an overhead of $(N/M)^\rho$, where $0<\rho<1$ is a parameter of the LSH employed, strictly improving both. 
We show:
\begin{theorem}\label{thm:main}
Consider $R,S \subseteq \mathbb{U}$, let $N = |R|+|S|$, assume $18 \log{N}+3B \leq M < N$ and that there exists a monotonic $(r,cr,p_1,p_2)$-sensitive family of functions with respect to distance measure $d$, using space $B$ and with $p_2 < p_1 < 1/2$. 
Let $\rho=\log{p_1}/\log{p_2}$.
Then there exists a cache-oblivious randomized algorithm computing $R \; \bowtie_{\leq r} S$ (w.r.t.~$d$) with probability $1-\BO{{1}/{N}}$ using
$$\TO{\left(\frac{N}{M}\right)^{\rho}
 \left(\frac{N}{B}+\frac{|R \underset{\leq r}{\bowtie} S|}{M
B}\right)
+ \frac{|R \underset{\leq cr}{\bowtie} S|}{M B} } \text{ I/Os.}\footnote{The $\TO{\cdot}$-notation hides polylog$(N)$ factors.}$$
\end{theorem}
%
We  conjecture that the bound in Theorem~\ref{thm:main} 
is close to the best possible for the class of ``signature based''
algorithms that work by generating a set of LSH values 
(from a black-box and monotonic family) and checking all pairs that collide. 
Our conjecture is based on an informal argument, given in full in Section~\ref{sec:exp}. 
We describe a worst-case input, where it seems significant advances are required to beat 
Theorem~\ref{thm:main} asymptotically. Further, we observe that for $M=N$ our bound coincides with
the optimal bound of reading the input, and when $M=1$ our bound coincides with the bounds of the best
known internal memory algorithms.

\smallskip

It is worth noting that whereas most methods in the literature focus on a single (or a few) distance measure, our method works for an arbitrary space and distance measure that allows LSH, e.g.,~Hamming, Manhattan~($\ell_1$), Euclidean~($\ell_2$), Jaccard, and angular metric distances. 
Since our approach makes use of LSH as a black box, the problem of reporting the complete join result with certainty would require major advances in LSH methods (see~\cite{Pacuk_COCOON16,Pagh_SODA16} for recent progress in this direction).

A primary technical hurdle in the paper is that we cannot use any kind of strong concentration bounds on the number of points having a particular value, since hash values of an LSH family may be correlated \emph{by definition}. 
Another hurdle is \emph{duplicate elimination} in the output stemming from pairs having multiple LSH collisions. 
However, in the context of I/O-efficient algorithms it is natural to not require the {\em listing\/} of all near pairs, but rather we simply require that the algorithm {\em enumerates\/} all such near pairs.
More precisely, the algorithm calls for each near pair $(x,y)$ a function \texttt{emit}$(x,y)$. 
This is a natural assumption in external memory since it reduces the I/O complexity. 
In addition, it is desired in many applications where join results are intermediate results pipelined to a subsequent computation, and are not required to be stored on external memory.
Our upper bound can be easily adapted to list all instances by increasing the I/O complexity of an \textit{unavoidable} additive term of $\BT{|R \; \bowtie_{\leq r} S|/B}$ I/Os.

\smallskip

\textbf{Organization.}
The organization of the paper is as follows. 
In Section~\ref{sec:RelatedWork}, we briefly review related work. 
Section~\ref{sec:OurAlgorithms} describes our algorithms including a warm-up cache-aware approach and the main results, a cache-oblivious solution, its analysis, and a randomized approach to remove duplicates. 
Section~\ref{sec:exp} provides some discussions on our
algorithms with some real datasets.
Section~\ref{sec:concl} concludes the paper.

\section{Related Work}\label{sec:RelatedWork}

In this section, we briefly review LSH, the computational I/O model, and some
state-of-the-art similarity join techniques.

\smallskip

\textbf{Locality-sensitive hashing (LSH).}
LSH was originally introduced by Indyk and Motwani~\cite{Indyk_STOC98} for
similarity search problems in high dimensional data. 
This technique obtains a sublinear (i.e.,~$\BO{N^{\rho}}$) time complexity by increasing the gap of collision probability between near points and far points using the LSH family as defined in Definition~\ref{def:LSH}. 
The gap of collision probability is polynomial, with an exponent of $\rho = \log{p_1}/\log{p_2}$ dependent on $c$.  

It is worth noting that the standard LSHs for metric distances, including Hamming~\cite{Indyk_STOC98}, $\ell_1$~\cite{Datar_SOCG04}, $\ell_2$~\cite{Andoni_FOCS06,Datar_SOCG04}, Jaccard~\cite{Broder_NETWORK97} and angular distances~\cite{Charikar_STOC02} are \textit{monotonic}. 
These common LSHs are space-efficient, and use space comparable to that required to store a point, except the LSH of~\cite{Andoni_FOCS06} which requires space $N^{o(1)}$.
We do not explicitly require the hash values themselves to be particularly small.
However, using universal hashing we can always map to small bit strings while introducing no new collisions with high probability. 
Thus we assume that $B$ hash values fit in one memory block.




\smallskip

\textbf{Computational I/O model.}
We study algorithms for similarity join in the \emph{external memory model}, which has been widely adopted in the literature (see, e.g., the survey by
Vitter~\cite{Vitter08}).
The external memory model consists of an internal memory of $M$ words and an external memory of unbounded size. 
The processor can only access data stored in the internal memory and move data between the two memories in blocks of size $B$.
For simplicity we will here measure block and internal memory size in units of points from $\mathbb{U}$, such that they can contain $B$ points and $M$ points, respectively.

The \emph{I/O complexity} of any algorithm is defined as the number of input/output blocks moved between the two memories by the algorithm. 
The \emph{cache-aware} approach makes explicit use of the parameters $M$ and $B$  to achieve its I/O complexity, whereas the \textit{cache-oblivious} one~\cite{frigo1999cache} does not explicitly use any model parameters. 
The latter approach is desirable as it implies optimality on all levels of the memory hierarchy and does not require parameter tuning when executed on different physical machines. 
Note that the cache-oblivious model assumes that the internal memory is \emph{ideal} in the sense that it has an optimal cache-replacement policy. Such cache-replacement policy can evict the block that is used furthest in the future, and can place a block anywhere in the cache (full associativity). 

\smallskip

\textbf{Similarity join techniques.}
We review some state-of-the-art of similarity join techniques most closely related to our work.

\begin{itemize}
	\item \textbf{Index-based similarity join.}
A popular approach is to make use of indexing techniques to build a data structure for one relation, and then perform queries using the points of the other relation.
The indexes typically perform some kind of \emph{filtering} to reduce the number of points that a given query point is compared to (see, e.g.,~\cite{Bayardo_WWW07,Chaudhuri_ICDE06,Gionis_VLDB99}).
Indexing can be space consuming, in particular for LSH, but in the context of similarity join this is not a big concern since we have many queries, and thus can afford to construct each hash table ``on the fly''.
On the other hand, it is clear that index-based similarity join techniques will not be able to take significant advantage of internal memory when $N\gg M$.
Indeed, the query complexity stated in~\cite{Gionis_VLDB99} is $\BO{(N/B)^\rho}$ I/Os.
Thus the I/O complexity of using indexing for similarity join will be high.
	\item \textbf{Sorting-based.}\label{sec:sorting}
The indexing technique of~\cite{Gionis_VLDB99} can be adapted to compute similarity joins more efficiently by using the fact that many points are being looked up in the hash tables.
This means that all lookups can be done in a batched fashion using sorting.
This results in a dependency on $N$ that is $\TO{(N/B)^{1 + \rho}}$ I/Os, where $\rho \in (0;1)$ is a parameter of the LSH family.
	\item \textbf{Generic joins.}\label{sec:generic}
When $N$ is close to $M$ the I/O-complexity can be improved by using general join operators optimized for this case.
It is easy to see that when $N/M$ is an integer, a nested loop join requires $N^2/(MB)$ I/Os.
Our cache-oblivious algorithm will make use of the following result on cache-oblivious nested loop joins:

\begin{theorem}(He and Luo~\cite{he2006cache})\label{thm:co-nested-loop}
	Given a similarity join condition, the join of relations $R$ and $S$ can be computed by a cache-oblivious algorithm in
	$$\BO{\frac{|R|+|S|}{B} + \frac{|R||S|}{MB}} \text{I/Os}.$$ 
		This number of I/Os suffices to generate the result in memory, but may not suffice to write it to disk.
\end{theorem}

We note that a similarity join can be part of a multi-way join involving more than two relations. For the class of \emph{acyclic joins}, where the variables compared in join conditions can be organized in a tree structure, one can initially apply a \emph{full reducer}~\cite{Yannakakis_VLDB81} that removes tuples that will not be part of the output. This efficiently reduces any acyclic join to a sequence of binary joins. Handling cyclic joins is much harder (see e.g.~\cite{Ngo_SIGMOD13}) and outside the scope of this paper.

\end{itemize}

\section{Our Algorithms}\label{sec:OurAlgorithms}

In this section we describe our I/O efficient algorithms.
We start in Section~\ref{sec:simple} with a warm-up cache-aware algorithm. It  uses an LSH family where the value of the collision probability is set to be a function of the internal memory size.
Section~\ref{sec:cacheobl} presents our main result, a recursive and cache-oblivious algorithm, which uses the LSH with a black-box approach and does not make any assumption on the value of collision probability. 
Section~\ref{sec:analysisIO} describes the analysis and Section~\ref{sec:remove-duplicates} shows how to reduce the expected number of times of emitting near pairs. 

\subsection{Cache-aware algorithm: \ASimJoin\ }\label{sec:simple}

We will now describe a simple cache-aware algorithm called \ASimJoin, which achieves the {worst case} I/O bounds as stated in Theorem~\ref{thm:main}. 
\ASimJoin\ relies on an $(r,cr,p'_1,p'_2)$-sensitive family $\mathcal H'$ of hash functions with the following properties: $p'_2 \leq M/N$ and $p'_1 \geq (M/N)^\rho$, for a suitable value $0 < \rho < 1$. 
Given an arbitrary monotonic $(r,cr,p_1,p_2)$-sensitive family $\mathcal H$, the family $\mathcal H'$ can be built by concatenating $\lceil \log_{p_2}(M/N)\rceil $ hash functions from $\mathcal H$. 
For simplicity, we assume that $\log_{p_2}(M/N)$ is an integer and thus the probabilities $p'_1$ and $p'_2$ can be exactly obtained. 
Nevertheless, the algorithm and its analysis can be extended to the general case by increasing the I/O complexity by a factor at most $p_1^{-1}$ in the worst case; in practical scenarios, this factor is a small constant~\cite{Broder_NETWORK97,Datar_SOCG04,Gionis_VLDB99}.

\SetCommentSty{}
\begin{algorithm}[!t]
\label{alg:ASimJoin}
\SetAlgoRefName{\ASimJoin$(R,S)$}
\SetKwSty{text}
\caption{$R, S$ are the input sets.}\label{algo:awaresimjoin}
\small

\SetKwBlock{K}{\textbf{Repeat} $3 \log{(N)}$ times}{}
\K{\label{step:logNrepeat}
	Associate to each point in $R$ and $S$ a counter initially set to $0$\;
	
	\SetKwBlock{R}{\textbf{Repeat} $L=2/p'_1$ times}{}
	\R{\label{step:Arepeat}
	
		Choose $h'_i\in\mathcal{H'}$ uniformly at random\label{step:A1}\;
		Use $h'_i$ to partition (in-place) $R$ and $S$ in buckets $R_v$, $S_v$ of points with the hash value~$v$\label{step:A2}\;

		\SetKwBlock{RR}{\textbf{For} each hash value $v$ generated in the previous step \label{step:A3}}{}
		\RR{
				\tcc{For simplicity we assume that $|R_v|\leq |S_v|$}
				Split $R_v$ and $S_v$ into chunks $R_{i,v}$ and $S_{i,v}$ of size at most $M/2$\;
				\SetKwBlock{F}{\textbf{For} every chunk $R_{i,v}$ of $R_v$}{}
				\F{
					Load in memory $R_{i,v}$\;
					\SetKwBlock{FF}{\textbf{For} every chunk $S_{i,v}$ of $S_v$ do}{}
					\FF{
						Load in memory $S_{i,v}$\;
						Compute $R_{i,v}\times  S_{i,v}$ and emit all near pairs. For each far pair, increment the associated counters by 1\;
						Remove from $S_{i,v}$ and
$R_{i,v}$ all points with the associated counter larger than $8LM$, and write
$S_{i,v}$ back to 	   	
						external
memory\label{step:remove}\;
						}		
					Write $R_{i,v}$ back to external memory\;
					}				
				}
			}
	}
\end{algorithm}

\ASimJoin\ assumes that each point in $R$ and $S$ is associated with a counter initially set to 0. 
This counter can be thought as an additional dimension of the point which hash functions and comparisons do not take into account.
The algorithm repeats $L=2/p'_1$ times the following procedure.
A hash function is randomly drawn from the $(r,cr,p'_1,p'_2)$-sensitive family, and it is used for partitioning the sets $R$ and $S$ into buckets of points with the same hash value.
We let $R_v$ and $S_v$ denote the buckets respectively containing points of $R$ and $S$ with the same hash value $v$.
Then, the algorithm iterates through every hash value and, for each hash value $v$, it uses a double nested loop for generating all pairs of points in $R_v\times S_v$.
The double nested loop loads consecutive chunks of $R_v$ and $S_v$ of size at most $M/2$: the outer loop runs on the smaller set (say $R_v$), while the inner one runs on the larger one (say $S_v$).
For each pair $(x,y)$, the algorithm emits the pair if $d(x,y)\leq r$,  increases by~1 counters associated with $x$ and $y$ if $d(x,y)>cr$,
or ignores the pair if  $r<d(x,y)\leq cr$.
Every time the counter of a point exceeds $8LM$, the point is considered to be
far away from all points and will be removed from the bucket. 
Chunks will be moved back in memory when they are no more needed. 
The entire \ASimJoin\ algorithm is repeated $3\log N$ times to find all
near pairs with high
probability. 
The following theorem shows the I/O bounds of the cache-aware approach.


\begin{theorem}\label{thm:aware}
Consider $R, S \subseteq \mathbb{U}$ and let $N = |R|+|S|$ be sufficiently large. 
Assume there exists a monotonic $(r,cr,p'_1,p'_2)$-sensitive family of functions with respect to distance measure $d$ with $p'_1=(M/N)^\rho$ and $p'_2 = M/N$, for a suitable value $0< \rho < 1$.
With probability $1-1/N$, the \ASimJoin\ algorithm enumerates all near pairs 
using
$$
\TO{ \left(\frac{N}{M}\right)^\rho \left(
\frac{N}{B} +
\frac{|R \underset{\leq r}{\bowtie} S|}{M
B}\right)
+ \frac{|R \underset{\leq cr}{\bowtie} S|}{M B} } \enspace \text{I/Os}.
$$
\end{theorem}

\begin{proof}

We observe that the I/O cost of Steps~\ref{step:A1}-\ref{step:A2}, that is of partitioning sets $R$ and $S$ according to a hash function $h'_i$, is $L \cdot \text{sort}(N) = \TO{(\frac{N}{M})^\rho \cdot \frac{N}{B}}$ for $L \leq (\frac{N}{M})^\rho$ repetitions\footnote{We let $\text{sort}(N) = \BO{(N/B) \log_{M/B} (N/B)}$ be shorthand for the I/O complexity \cite{Vitter08} of sorting $N$ points.}. 

We now consider the I/O cost of an iteration of the loop in Step~\ref{step:A3} for a given hash value $v$. 
When the size of one bucket (say $R_v$) is smaller than $M/2$, we are able to load the whole $R_v$ into the internal memory and then load consecutive blocks of $S_v$ to execute join operations. 
Hence, the I/O cost of this step is at most $(|R_v|+|S_v|)/B$.
The total I/O cost of the $3L\log{(N)}$ iterations of Step~\ref{step:A3} among
all possible hash values where at least one bucket has size smaller than $M/2$
is at most $L \cdot \frac{N}{B} = 2(\frac{N}{M})^\rho \cdot \frac{N}{B}$  I/Os.

The I/O cost of Step~\ref{step:A3} when both buckets $R_v$ and $S_v$ are larger than $M/2$ is $2|R_v||S_v|/(BM)$. 
This means that the amortized cost of each pair in $R_v\times S_v$ is $2/(BM)$. 
Therefore the amortized I/O cost of all iterations of Step~\ref{step:A3}, when there are no bucket size less than $M/2$, can be upper bounded by multiplying the total number of generated pairs by $2/(BM)$.
Based on this observation, we classify and enumerate generated pairs into three groups: near pairs, $c$-near pairs and far pairs. We denote by $C_n$, $C_{cn}$ and $C_f$ the respective size of each group, and upper bound these quantities to derive the proof.
\begin{enumerate}

\item \emph{Number of near pairs.} By definition, LSH gives a lower bound on the probability of collision of near pairs. 
It may happen that the collision probability of near pairs is~1. 
Thus, two near points might collide in all $L$ repetitions of
Step~\ref{step:Arepeat} and in all $3\log{(N)}$ repetitions of
Step~\ref{step:logNrepeat}. 
This means that $C_n \leq 3L\log{(N)}|R \bowtie_{\leq r} S|$. Note that this
bound is a deterministic worst case bound.

\item \emph{Number of $c$-near pairs}. Any $c$-near pair from $R \bowtie_{\leq cr} S$ appears in a bucket with probability at most $p'_1$ due to monotonicity of our LSH family. 
Since we have $L=2/p'_1$ repetitions, each $c$-near pair collides at most~2 in
expectation. 
In other words, the expected number of $c$-near pair collisions among $L$
repetitions is at most $2|R \bowtie_{\leq cr} S|$.
By using the Chernoff bound~\cite[Exercise 1.1]{PanconesiDubhashiBook} with
$3\log{(N)}$ independent $L$ repetitions (in Step~\ref{step:logNrepeat}), we
have 

\begin{displaymath}
\begin{aligned}
\Pr{C_{cn} \geq 6\log{(N)}|R \bowtie_{\leq cr} S| } &\leq 1/N^2, \\
\Pr{C_{cn} \leq 6\log{(N)}|R \bowtie_{\leq cr} S| } &\geq 1 - 1/N^2.
\end{aligned}
\end{displaymath}

\item \emph{Number of far pairs}. If $x \in R \cup S$ is far away from all
points, the expected number of collisions of $x$ in $L$ hash table (including
duplicates) is at most $8LM$, since then the point is removed by
Step~\ref{step:remove}.
Hence the total number of examined far pairs is $C_f \leq
24NLM\log{(N)}$.

\end{enumerate}
Therefore, by summing the number of near pairs $C_n$, $c$-near pairs $C_{cn}$, and far pairs $C_f$, and multiplying these quantities by the amortized I/O complexity $2/(BM)$, we upper bound the I/O cost of all iterations of Step~\ref{step:A3}, when there are no buckets of size less than $M/2$, is
$$
\TO{\left(\frac{N}{M}\right)^\rho \left(\frac{N}{B}+ \frac{|R \bowtie_{\leq r} S|}{BM}\right)+\frac{|R \bowtie_{\leq cr} S|}{BM}},
$$
with probability at least $1-1/N^2$.
By summing all the previous bounds, we get the claimed I/O bound with high probability.

%
%
%
%

We now  analyze the probability to enumerate all near pairs.
Consider one iteration of Step~\ref{step:logNrepeat}. A near pair is not emitted
if at least one of the following events happen:
\begin{enumerate}
 \item The two points do not collide in the same bucket in each of the $L$
iterations of Step~\ref{step:Arepeat}. This happens with probability
$(1-p'_1)^L = (1-p'_1)^{2/p'_1} \leq 1/e^2$.
 \item One of the two points is removed by Step~\ref{step:remove} because it
collides with more than $8 LM$ far points.
By the Markov's inequality and since there are at most $N$ far points, the
probability that $x$ collides with at least $8 LM$ points in the $L$ iterations
is at most $1/8$. Then, this event happens with probability at most $1/4$.
\end{enumerate}
Therefore, a near pair does not collide in one  iteration of
Step~\ref{step:logNrepeat} with probability at most $1/e^2+1/4<1/2$ and never
collides in the $3\log N$ iterations with probability at most $(1/2)^{3 \log
N}=1/N^{3}$.
Then, by an union bound, it follows that all near pairs (there are at most $N^2$
of them) collide with probability at least $1-1/N$ and the theorem
follows.
%
%
\qed
\end{proof}

As already mentioned in the introduction, a near pair $(x,y)$ can be emitted many times during the algorithm since points $x$ and $y$ can be hashed on the same value in $p(x,y)L$ rounds of Step~\ref{step:Arepeat}, where $p(x,y)\geq p'_1$ denotes the actual collision probability.
A simple approach for avoiding duplicates is the following: for each near pair found during the $i$-th iteration of Step~\ref{step:Arepeat}, the pair is emitted only if the two points did not collide by all hash functions used in the previous $i-1$ rounds.
The check starts from the hash function used in the previous round and backtracks until a collision is found or there are no more hash functions. This approach increases the worst case complexity by a factor $L$.
Section~\ref{sec:remove-duplicates} shows a more efficient randomized algorithm that reduces the number of replica per near pair to a constant. 
This technique also applies to the cache-oblivious algorithm described in the next section.

\subsection{Cache-oblivious algorithm: \SimJoin\ }\label{sec:cacheobl}

The above cache-aware algorithm uses an $(r,cr, p'_1, p'_2)$-sensitive family of functions $\mathcal H'$, with  $p'_1\sim (M/N)^\rho$ and $p'_2\sim M/N$, for partitioning the initial sets into smaller buckets, which are then efficiently processed in the internal memory using the nested loop algorithm. 
If we know the internal memory size $M$, this LSH family can be constructed by concatenating $\lceil \log_{p_2} (M/N)\rceil$ hash functions from any given primitive $(r,cr,p_1,p_2)$-sensitive family $\mathcal H$. Without knowing $M$ in the cache-oblivious setting, such family cannot be built. Therefore, we propose \SimJoin, a cache-oblivious algorithm that efficiently computes the similarity join without knowing the internal memory size $M$ and the block length $B$.

\SimJoin\ uses as a black-box a given monotonic $(r,cr, p_1, p_2)$-sensitive family $\mathcal H$.\footnote{The monotonicity requirement can be relaxed to the following: $\Pr{h(x)=h(y)} \geq \Pr{h(x')=h(y')}$ for every two pairs $(x,y)$ and $(x',y')$ where $d(x,y)\leq r$ and $d(x',y')>r$. A monotonic LSH family clearly satisfies this assumption.}
The value of $p_1$ and $p_2$ can be considered constant in a practical scenario.
As common in cache-oblivious settings, we use a recursive approach for splitting the problem into smaller and smaller subproblems that at some point will fit the internal memory, although this point is not known in the algorithm.
We first give a high level description of the cache-oblivious algorithm and an intuitive explanation. 
We then provide a more detailed description and analysis.

\begin{algorithm}[!t]
\SetAlgoRefName{\SimJoin$(R,S, \psi)$}
\SetKwSty{text}
\caption{$R,S$ are the input sets, and $\psi$ is the recursion depth.}\label{algo:simjoin}
\small

\textbf{If} $|R|>|S|$, \textbf{then} swap (the references to) the sets such that $|R|\leq
|S|$\label{step:begin}\;

\textbf{If} $\psi=\Psi$ or $|R|\leq 1$, \textbf{then} compute $R \bowtie_{\leq r} S$ using the
algorithm of Theorem~\ref{thm:co-nested-loop} and return\label{step:basecase}\;

Pick a random sample $S'$ of $18\Delta$ points from $S$ (or all points if $|S|<18\Delta$)\label{step:sample}\;

Compute $R'$ containing all points of $R$ that have distance smaller than $cr$ to at least half points in $S'$\label{step:estimate}\;  

Compute $R' \bowtie_{\leq r} S$ using the algorithm of Theorem~\ref{thm:co-nested-loop}\label{step:nested-loop}\;

\SetKwBlock{R}{\textbf{Repeat} $L=1/p_1$ times\label{step:repeat}}{}
\R{
 Choose $h\in\mathcal{H}$ uniformly at random\label{step:repeat1}\;
 Use $h$ to partition (in-place) $R\backslash R'$ and $S$ in buckets $R_v$, $S_v$ of points with hash value~$v$\label{step:repeat2}\;
 \textbf{For} each $v$ where $R_v$ and $S_v$ are nonempty, recursively call \newline {\sc \SimJoin\ }$(R_v,S_v, \psi+1)$\label{step:repeat3}\;
}

\end{algorithm}

\SimJoin\ receives in input the two sets $R$ and $S$ of similarity join, and a parameter $\psi$ denoting the depth in the recursion tree (initially, $\psi=0$) that is used for recognizing the base case. 
Let  $|R| \leq |S|$, $N=|R|+|S|$, and denote with $\Delta=\log N$ and $\Psi = \lceil \log_{1/p_2} N\rceil$ two global values that are kept invariant in the recursive levels and computed using the initial input size $N$. 
For simplicity we assume that $1/p_1$ and $1/p_2$ are integers, and further assume without loss of generality that the initial size $N$ is a power of two. 
Note that, if $1/p_1$ is not an integer, the last iteration in Step~\ref{step:repeat} can be performed with a random variable $L\in \{\lfloor 1/p_1\rfloor,\lceil 1/p_1\rceil\}$ such that $\E{L}=1/p_1$.

\SimJoin\ works as follows. 
If the problem is currently at the recursive level $\Psi = \lceil \log_{1/p_2} N\rceil$ or $|R|\leq 1$, the recursion ends and  the problem is solved using the cache-oblivious nested loop described in Theorem~\ref{thm:co-nested-loop}. 
Otherwise, the following operations are executed.
By exploiting sampling, the algorithm identifies a subset $R'$ of $R$ containing (almost) all points that are near or $c$-near to a constant fraction of points in $S$ (Steps~\ref{step:sample} -- \ref{step:estimate}).
Then we compute $R' \bowtie_{\leq r} S$ using the cache-oblivious nested-loop of Theorem~\ref{thm:co-nested-loop} and remove points in $R'$ from $R$ (Step~\ref{step:nested-loop}).
Subsequently, the algorithm repeats $L=1/p_1$ times the following operations: a hash function is extracted from the $(r,cr, p_1, p_2)$-sensitive family and used for partitioning $R$ and $S$ into buckets, denoted with $R_v$ and $S_v$ with any hash value $v$ (Steps~\ref{step:repeat1} -- \ref{step:repeat2});
then, the join $R_v \bowtie_{\leq r} S_v$ is computed recursively by \SimJoin (Step~\ref{step:repeat3}).

The explanation of our approach is the following. 
By recursively partitioning input points with hash functions from $\mathcal H$, the algorithm decreases the probability of collision between two far points. 
In particular, the collision probability of a far pair is $p_2^i$ at the $i$-th recursive level.
On the other hand, by repeating the partitioning $1/p_1$ times in each level, the algorithm  guarantees that a near pair is enumerated with constant probability since the probability that a near pair collide is $p_1^i$ at the $i$-th recursive level.
It deserves to be noticed that the collision probability of far and near pairs at the recursive level $\log_{1/p_2} (N/M)$ is $\BT{M/N}$ and $\BT{(M/N)^\rho}$, respectively, which are asymptotically equivalent to the values in the cache-aware algorithm. 
In other words, the partitioning of points at this level is equivalent to the one in the cache-aware algorithm with collision probability for a far pair $p_2' = M/N$.
Finally, we observe that, when a point in $R$ becomes close to many points in $S$, it is more efficient to detect and remove it, instead of propagating it down to the base cases. 
This is due to the fact that the collision probability of very near pairs is always large (close to~1) and the algorithm is not able to split them into subproblems that fit in memory.




\subsection{I/O Complexity and Correctness of \SimJoin\ }\label{sec:analysisIO}

\subsubsection{Analysis of I/O Complexity.}

We will bound the \emph{expected} number of I/Os of the algorithm rather than the worst case. 
This can be converted to an high probability bound by running $\log N$ parallel instances of our algorithm 
(without loss of generality we assume that the optimal cache replacement 
splits the cache into $M/\log N$ parts that are assigned
to each instance).
The total execution stops when the first parallel instance terminates, which with probability at least 
$1 - 1/N$ is within a logarithmic factor of the expected I/O bound (logarithmic factors 
are absorbed in the $\tilde{\Osymbol}$-notation).

For notational simplicity, in this section we let $R$ and $S$ denote the initial input sets and let $\tilde R$ and $\tilde S$ denote the subsets given in input to a particular recursive subproblem (note that, due to Step~\ref{step:begin}, $\tilde R$  can denote a subset of $R$ but also of $S$; similarly for $\tilde S$). 
We also let $\tilde S'$ denote the sampling of $\tilde S$ in Step~\ref{step:sample}, and with $\tilde R'$ the subset of $\tilde R$ computed in Step~\ref{step:estimate}. 
Lemma~\ref{lemma:sample} says that two properties of the choice of random sample in Step~\ref{step:sample} are almost certain, and the proof relies on Chernoff bounds on the choice of $\tilde S'$. 
In the remainder of the paper, we assume that Lemma~\ref{lemma:sample} holds and refer to this event as $\mathcal{A}$ holding with probability $1-\BO{1/N}$.

\begin{lemma}\label{lemma:sample}
With probability at least $1-\BO{1/N}$ over the random choices in Step~\ref{step:sample},
the following bounds hold for every  subproblem \SimJoin$(\tilde R,\tilde S,\psi)$:
%
\begin{equation}\label{eq:manyclose}
|\tilde R' \underset{\leq cr}{\bowtie} \tilde S| > \frac{|\tilde R'| |\tilde
S|}{6} \enspace,
\end{equation}
\begin{equation}\label{eq:manyfar}
 |(\tilde R\backslash \tilde R') \underset{> cr}{\bowtie} \tilde S| >
\frac{|\tilde R\backslash \tilde R'| |\tilde S|}{6} \enspace .
\end{equation}
\end{lemma}

\begin{proof}


Let $x\in \tilde R$ be a point which is $c$-near to at most one sixth of the points in $\tilde S$, 
i.e.,~$|x \enspace {\bowtie}_{\leq cr} \tilde S| \leq |\tilde S|/6$.
The point $x$ enters $\tilde R'$ if there are at least $9\Delta$ $c$-near points in 
$\tilde S'$ and this happens, for a 
Chernoff bound~\cite[Theorem 1.1]{PanconesiDubhashiBook},  with probability at most $1/N^4$.
Each point of $R\cup S$  appears in at most $2 \sum_{i=0}^{i=\Psi - 1} L^i < 2 L^\Psi < 2N^2$ subproblems 
and there are at most $N$ points in $R\cup S$.
Therefore, with probability $1-2N^3N^{-4} = 1-2N^{-1}$, we have that 
in every subproblem \SimJoin$(\tilde R,\tilde S,\psi)$ no point with at most $|\tilde S|/6$ $c$-near 
points in $\tilde S$ is in $\tilde R'$. Hence each point in $\tilde R$ has at least  $|\tilde S|/6$ $c$-near 
points in $\tilde S$, and the bound in Equation~\ref{eq:manyclose} follows.

We can similarly show that,  with probability $1-2N^3N^{-4} = 1-2N^{-1}$, we have that 
in every subproblem \SimJoin$(\tilde R,\tilde S,\psi)$ all points with at least  $5|\tilde S|/6$ $c$-near 
points in $\tilde S$ are in $\tilde R'$. Then, each point in $\tilde R\backslash \tilde R'$ has $ |\tilde S|/6$
far points in $\tilde S$ and Equation~\ref{eq:manyfar} follows.
\qed

\end{proof}

To analyze the number of I/Os for subproblems of size more than $M$ we bound the cost in terms of different types of \emph{collisions} of pairs in $R\times S$ that end up in the same subproblem of the recursion. 
We say that $(x,y)$ \emph{is in} a particular subproblem \SimJoin$(\tilde{R},\tilde{S}, \psi)$ if $(x,y)\in (\tilde{R}\times \tilde{S}) \cup (\tilde{S}\times \tilde{R})$. 
Observe that a pair $(x,y)$ is in a subproblem if and only if $x$ and $y$ have colliding hash values on every step of the call path from the initial invocation of \SimJoin.

\newcommand{\C}[2][i]{C_{#1}\left({#2} \right)}
\begin{definition}

Given $Q\subseteq R\times S$ let $\C{Q}$ be the number of times a pair in $Q$ is in a call to \SimJoin\ at the $i$-th level of recursion. 
We also let $\C[i,k]{Q}$, with $0\leq k \leq \log M$, denote the number of times a pair in $Q$ is in a call to \SimJoin\ at the $i$-th level of recursion where the smallest input set has size in $[2^k, 2^{k+1})$ if $0\leq k <\log M$, and in $[M, +\infty)$ if $k = \log M$. 
The count is over all pairs and with multiplicity, so if $(x,y)$ is in several subproblems at the $i$-th level, all these are counted.

\end{definition}

Next we bound the I/O complexity of \SimJoin\ in terms of $\C{R \bowtie_{\leq cr} S}$ and $\C[i,k]{R \bowtie_{> cr} S}$, for any $0\leq i < \Psi$. 
We will later upper bound the expected size of these quantities in Lemma~\ref{lemma:expectation} and then get the claim of Theorem~\ref{thm:main}. 

\begin{lemma}\label{lem:collision-bound}

Let $\ell=\lceil \log_{1/p_2} (N/M)\rceil$ and $M\geq 18\log N + 3B$. Given that $\mathcal A$ holds,  the I/O complexity of \SimJoin$(R,S,0)$ is
\[
\TO{\frac{NL^{\ell}}{B}+ \sum_{i=0}^\ell \frac{\C{R \underset{\leq
cr}{\bowtie} S}}{MB}  + \sum_{i=\ell}^{\Psi-1}\sum_{k=0}^{\log
M}\frac{\C[i,k]{ R
\underset{> cr}{\bowtie}  S} L }{B 2^k}}\]

\end{lemma}

\begin{proof} 


To ease the analysis we assume that no more than 1/3 of internal memory is used to store blocks containing elements of $R$ and $S$, respectively. 
Since the cache-oblivious model assumes an optimal cache replacement policy this cannot decrease the I/O complexity. 
Also, internal memory space used for other things than data (input and output buffers, the recursion stack of size at most $\Psi$) is less than $M/3$ by our assumption that $M\geq 18\log n+3B$. 
As a consequence, we have that the number of I/Os for solving a subproblem \SimJoin\ $(\tilde{R},\tilde{S}, \cdot)$ where $|\tilde{R}|\leq
M/3$ and $|\tilde S|\leq M/3$ is $\BO{(|\tilde{R}|+|\tilde{S}|)/B}$, including all recursive calls. 
This is because there is space $M/3$ dedicated to both input sets and only I/Os for reading the input are required. 
By charging the cost of such subproblems to the writing of the inputs in the parent problem, we can focus on subproblems where the largest set (i.e., $\tilde{S}$) has size more than $M/3$.
We notice that the cost of Steps~\ref{step:sample}--\ref{step:estimate} is dominated by other costs by our assumption that the set $\tilde S'$ fits in internal memory, which implies that it suffices to scan data once to implement these steps. 
This cost is clearly negligible with respect to the remaining steps and thus we ignore them.

We first provide an upper bound on the I/O complexity required  by all subproblems at a recursive level above $\ell$. 
Let \SimJoin\ $(\tilde{R},\tilde{S},i)$ be a recursive call at the $i$-th recursive level, for $0 \leq i \leq \ell$. 
The I/O cost of the nested loop join in Step~\ref{step:nested-loop} in \SimJoin\ $(\tilde{R},\tilde{S},i)$ is $\BO{|\tilde{S}|/B + |\tilde{R}'| |\tilde{S}| / (MB)}$ by Theorem~\ref{thm:co-nested-loop}. 
We can ignore the $\BO{|\tilde{S}|/B}$ term since it is asymptotically negligible with respect to the cost of each iteration of Step~\ref{step:repeat}, which is upper bounded later. 
By Equation~\ref{eq:manyclose}, we have that $\tilde{R}' \bowtie_{\leq cr} \tilde{S}$ contains more than
$|\tilde{R}'| |\tilde{S}|/6$ pairs, and thus the cost of Step~\ref{step:nested-loop} in \SimJoin$(\tilde{R},\tilde{S},i)$ is $\BO{|\tilde
R'\bowtie_{\leq cr} \tilde S|/(MB)}$. 
This means that we can bound the total I/O cost of all executions of Step~\ref{step:nested-loop} at level $i$ of the recursion with $\BO{\C{R \bowtie_{\leq cr} S}/(MB)}$ since each near pair $(x,y)$ appears in $C_i((x,y))$  subproblems at level $i$.

The second major part of the I/O complexity is the cost of preparing recursive calls in \SimJoin$(\tilde{R},\tilde{S},i)$ (i.e., Steps~\ref{step:repeat1}--\ref{step:repeat2}). 
In fact, in each iteration of Step~\ref{step:repeat}, the I/O cost is $\TO{(|\tilde R|+|\tilde S|)/B}$, which includes the cost of hashing and of sorting to form buckets. 
Since each point of $\tilde R$ and $\tilde S$ is replicated in $L$ subproblems in Step~\ref{step:repeat}, we have that each point of the initial sets $R$ and $S$ is replicated $L^{i+1}$ times at level $i$. 
Since the average cost per entry is $\TO{1/B}$, we have that the total cost for preparing  recursive calls at level
$i$ is $\TO{N L^{i+1}/B}$.
By summing the above terms, we have that the total I/O complexity of all subproblems in the $i$-th recursive level is upper bounded by:

\begin{equation}\label{eq:belowL}
\TO{\frac{\C{R \underset{\leq cr}{\bowtie} S}}{MB} + \frac{NL^{i+1}}{B}}.
\end{equation}

We now focus our analysis to  bound the I/O complexity required by all subproblems at a recursive level below $\ell$.
Let again \SimJoin$(\tilde{R},\tilde{S},i)$ be a recursive call at the $i$-th recursive level, for $\ell \leq i\leq \Psi$.
We observe that (part of) the cost of a subproblem at level $i \geq \ell$ can be upper bounded by a suitable function of collisions among far points in
\SimJoin\ $(\tilde{R},\tilde{S},i)$.
More specifically, consider an iteration of Step~\ref{step:repeat} in a subproblem at level $i$. 
Then, the cost for preparing the recursive calls and for performing Step~\ref{step:nested-loop} in each  subproblem (at level $i+1$) generated during the iteration, can be upper bounded as
\[
\TO{\frac{|\tilde R\backslash \tilde R'|+|\tilde S|}{B}+\frac{|(\tilde
R\backslash \tilde R') \bowtie_{\leq cr}\tilde S|}{BM}}\text{,}\]
since each near pair in $(\tilde R\backslash \tilde R') \bowtie_{\leq cr} \tilde S$ is found in Step~\ref{step:nested-loop} in at most one subproblem at level $i+1$ generated during the iteration. 
Since we have that $|(\tilde R\backslash \tilde R') \bowtie_{\leq cr}\tilde S|\leq |\tilde R\backslash \tilde R'| |\tilde S|$, we easily get that the above bound can be rewritten as $\TO{{|\tilde R\backslash \tilde R'||\tilde S|}/({B \min\{M,|\tilde R\backslash \tilde R'|\})}}$. 
We observe that this bound holds even when $i=\Psi-1$: in this case the cost includes all I/Os required for solving the subproblems at level $\Psi$ called in the iteration and which are solved using the nested loop in Theorem~\ref{thm:co-nested-loop} (see Step~\ref{step:basecase}). 
By Lemma~\ref{lemma:sample}, we have that the above quantity can be upper bounded with the number of far collisions between $\tilde R$ and $\tilde S$, getting $\TO{({|\tilde R\backslash \tilde R' \bowtie_{>cr} \tilde S|})/({B\min\{M,|\tilde R\backslash \tilde R'|\})}}$. 

Recall that $\C[i,k]{Q}$ denotes the number of times a pair in $Q$ is in a call to \SimJoin\ at the $i$-th level of recursion where the smallest input set has size in $[2^k, 2^{k+1})$ if $0\leq k <\log M$, and in $[M, +\infty)$ if $k = \log M$. 
Then, the total cost for preparing the recursive calls in Steps~\ref{step:repeat1}--\ref{step:repeat2} in all subproblems at level $i$ and for performing Step~\ref{step:nested-loop} in all subproblems at level $(i+1)$ is:\footnote{We note that the true input size of a subproblem is $|\tilde R|$
and not $|\tilde R\backslash \tilde R'|$. However, the expected value of $\C[i,k]{R \bowtie_{> cr} S}$ is computed assuming the worst case where there are no close pairs an thus $\tilde R'=\emptyset$.}
\begin{equation}\label{eq:aboveL}
\TO{\sum_{k=0}^{\log M}\frac{\C[i,k]{R \bowtie_{> cr} S} L}{B 2^k}}.
\end{equation}
The $L$ factor in the above bound follows since far collisions at level $i$ are used for amortizing the cost of Step~\ref{step:nested-loop} for each one of the $L$ iterations of Step~\ref{step:repeat}.

To get the total I/O complexity of the algorithm we sum the I/O complexity required by each recursive level. 
We bound the cost of each level as follows:
for a level $i<\ell$ we use the bound in Equation~\ref{eq:belowL}; 
for a level $i>\ell$ we use the bound in Equation~\ref{eq:aboveL}; 
for level $i=\ell$, we use the bound given in Equation~\ref{eq:aboveL} to which we add the first term in Equation~\ref{eq:belowL} since  the cost of
Step~\ref{step:nested-loop} at level $\ell$ is not included in Equation~\ref{eq:aboveL} (note that the addition of Equations~\ref{eq:belowL}
and~\ref{eq:aboveL} gives a weak upper bound for level $\ell$).
The lemma follows.\qed
\end{proof}

We will now analyze the expected sizes of the terms in Lemma~\ref{lem:collision-bound}. 
Clearly each pair from $R\times S$ is in the top level call, so the number of collisions is $|R||S| < N^2$. 
But in lower levels we show that the expected number of times that a pair collides either decreases or increases geometrically, depending on whether the collision probability is smaller or larger than $p_1$ (or equivalently, depending on whether the distance is greater or smaller than the radius $r$). 
The lemma follows by expressing the number of collisions of the pairs at the $i$-th recursive level as a \emph{Galton-Watson branching process}~\cite{harris2002theory}. 

\begin{lemma}\label{lemma:expectation}
Given that $\mathcal{A}$ holds, for each $0\leq i \leq \Psi$ we have

\begin{enumerate}
	\item $\E{\C{R \underset{> cr}{\bowtie} S}} \leq |R \underset{>
cr}{\bowtie} S|\, (p_2/p_1)^i$\label{eq:numfar};
	\item $\E{\C{R \underset{> r, \leq cr}{\bowtie} S}} \leq |R \underset{>r, \leq
cr}{\bowtie} S|$ \label{eq:numcnear};
	\item $\E{\C{R \underset{\leq r}{\bowtie} S}} \leq |R \underset{\leq
r}{\bowtie} S|\, L^i$\label{eq:numnear};
	\item $\E{\C[i,k]{R \underset{> cr}{\bowtie} S}} \leq N 2^{k+1}\,
(p_2/p_1)^i$, for any $0\leq k <\log M$\label{eq:numfarK}.
\end{enumerate}
\end{lemma}

\begin{proof}

Let $x\in R$ and $y\in S$. 
We are interested in upper bounding the number of collisions of the pair at the $i$-th recursive level. 
We envision the problem as \emph{branching process} (more specifically a Galton–Watson process, see e.g.~\cite{harris2002theory}) where the expected number of children (i.e., recursive calls that preserve a particular collision) is $\Pr{h(x)=h(y)}/p_1$ for random $h\in\mathcal{H}$. 
It is a standard fact from this theory that the expected population size at generation $i$ (i.e., number of times $(x,y)$ is in a problem at recursive level $i$) is $(\Pr{h(x)=h(y)}/p_1)^i$~\cite[Theorem 5.1]{harris2002theory}. 
If $d(x,y)>cr$, we have that $\Pr{h(x)=h(y)}\leq p_2$ and each far pair appears at most $(p_2/p_1)^i$ times in expectation at level $i$, from which follows Equation~\ref{eq:numfar}. 
Moreover, since the probability of collisions is monotonic in the distance, we have that  $\Pr{h(x)=h(y)}\leq 1$ if $d(x,y)\leq r$, and $\Pr{h(x)=h(y)}\leq p_1$ if $r<d(x,y)\leq cr$, from which follow Equations~\ref{eq:numcnear} and~\ref{eq:numnear}.

In order to get the last bound we observe that each entry of $R$ and $S$ is replicated $L^i=p_1^{-i}$ times at level $i$. 
Thus, we have that $N 2^{k+1} L^i$ is the total maximum number of far collisions in subproblems at level $i$ where the smallest input set has size in $[2^k, 2^{k+1})$. 
Each one of these collisions survives up to level $i$ with probability $p_2^i$, and thus the expected number of these collisions is $N 2^{k+1} (p_1/p_2)^i$. \qed

\end{proof}

We are now ready to prove the I/O complexity of \SimJoin\ as claimed in Theorem~\ref{thm:main}.
By the linearity of expectation and Lemma~\ref{lem:collision-bound}, we get that the expected I/O complexity of \SimJoin\ is
$$
\TO{\frac{NL^{\ell}}{B} + \sum_{i=0}^\ell \frac{\E{\C{R \underset{\leq
cr}{\bowtie} S}}}{MB}
+ \sum_{i=\ell}^{\Psi-1}\sum_{k=0}^{\log
M}\frac{\E{\C[i,k]{ R \underset{> cr}{\bowtie}  S}} L
}{B 2^k}} \enspace,
$$
where $\ell=\lceil \log_{1/p_2}(N/M)\rceil$. 
Note that $\C[i,\log M]{ R \bowtie_{> cr} S}\leq \C{ R \bowtie_{> cr}  S}$ we have $|R \bowtie_{> cr}  S|\leq N^2$ and
$\C{R \bowtie_{\leq cr} S}=\C{R \bowtie_{\leq r} S}+\C{R \bowtie_{> r, \leq cr} S}$. By plugging in the
bounds on the expected number of collisions given in Lemma~\ref{lemma:expectation}, we get the claimed result.


\subsubsection{Analysis of Correctness.}\label{sec:correctness}
The following lemma shows that \SimJoin\ outputs with probability $1 - \BO{1/N}$
all  near pairs,  as claimed in Theorem~\ref{thm:main}.
\begin{lemma}\label{lem:prob}
Let $R,S \subseteq \mathbb{U}$ and $|R| + |S| = N$. Executing $\BO{\log^{3/2} N}$ independent 
repetitions of \SimJoin(R,S,0) outputs $R \bowtie_{\leq r}  S$ with probability at least $1 - \BO{1/N}$.
\end{lemma}
\begin{proof}
We now argue that a pair $(x,y)$ with $d(x,y)\leq r$ is output with probability $\Omega ( 1 / \sqrt{\log N})$.
Let $X_i = C_i((x,y))$ be the number of subproblems at the level $i$ containing $(x,y)$.
By applying Galton-Watson branching process, we get that $\E{X_i}=(\Pr{h(x)=h(y)}/p_1)^i$.
If $\Pr{h(x)=h(y)}/p_1 > 1$ then in fact there is positive constant probability that $(x,y)$ survives indefinitely, i.e., does not go extinct~\cite{harris2002theory}.
Since at every branch of the recursion we eventually compare points that collide under all hash functions on the path from the root call, this implies that $(x,y)$ is reported with a positive constant probability.

In the \emph{critical case} where $\Pr{h(x)=h(y)}/p_1 = 1$ we need to consider the variance of $X_i$, which by~\cite[Theorem 5.1]{harris2002theory} is equal to $i\sigma^2$, where $\sigma^2$ is the variance of the number of children (hash collisions in recursive calls).
If $1/p_1$ is integer, the number of children in our branching process follows a binomial distribution with mean~1.
This implies that $\sigma^2 < 1$.
Also in the case where $1/p_1$ is not integer, it is easy to see that the variance is bounded by 2. 
That is, we have $\V{X_i} \leq 2i$, which by Chebychev's inequality means that for some integer $j^* = 2\sqrt{i} + \BO{1}$:
$$\sum_{j=j^*}^\infty \Pr{X_i \geq j} \leq \sum_{j=j^*}^\infty \V{X_i}/j^2 \leq 1/2 \enspace . $$
Since we have $\E{X_i} = \sum_{j=1}^\infty \Pr{X_i \geq j} = 1$ then $\sum_{j=1}^{j^*-1} \Pr{X_i \geq j} > 1/2$, and since $\Pr{X_i \geq j}$ is non-increasing with $j$ this implies that $\Pr{X_i \geq 1} \geq 1/(2j^*) = \BOM{1/\sqrt{i}}$.
Furthermore, the recursion depth $\BO{\log N}$ implies the probability that a near pair is found is $\BOM{1/\sqrt{\log N}}$.
Thus, by repeating $\BO{\log^{3/2} N}$ times we can make the error probability $\BO{1/N^3}$ for a particular pair and $\BO{1/N}$ for the entire output by applying the union bound.
\end{proof}

\subsection{Removing duplicates}\label{sec:remove-duplicates}

Given two near points $x$ and $y$, the definition of LSH requires their collision probability $p(x,y)=\Pr{h(x)=h(y)} \geq p_1$.
If $p(x,y) \gg p_1$, our \SimJoin\ algorithm can emit $(x,y)$ many times. 
As an example suppose that the algorithm ends in one recursive call: then, the pair $(x,y)$ is expected to be in the same bucket for $p(x,y) L$ iterations of Step~\ref{step:repeat} and thus it is emitted $p(x,y) L \gg 1$ times in expectation. 
Moreover, if the pair is not emitted in the first recursive level, the expected number of emitted pairs increases as $(p(x,y) L)^i$ since the pair $(x,y)$ is contained in $(p(x,y)L)^i$ subproblems at the $i$-th recursive level.
A simple solution requires to store all emitted near pairs on the external memory, and then using a cache-oblivious sorting algorithm~\cite{frigo1999cache} for removing repetitions. 
However, this approach requires $\TO{\kappa \frac{|R \bowtie_{\leq r} S|}{B}}$ I/Os, where $\kappa$ is the expected average replication of each emitted pair, which can dominate the complexity of \SimJoin.
A similar issue appears in the cache-aware algorithm \ASimJoin\ as well: a near pair is emitted at most $L'=(N/M)^\rho$ times since there is no recursion and the  partitioning of the two input sets is repeated only $L'$ times.

If the collision probability $\Pr{h(x)=h(y)}$ can be explicitly computed in $\BO{1}$ time and no I/Os for each pair $(x,y)$, it is possible to emit each near pair once in expectation without storing near pairs on the external memory. 
We note that the collision probability can be computed for many metrics, including Hamming~\cite{Indyk_STOC98}, $\ell_1$ and $\ell_2$~\cite{Datar_SOCG04},
Jaccard~\cite{Broder_NETWORK97}, and angular~\cite{Charikar_STOC02} distances. 
For the cache-oblivious algorithm, the approach is the following: for each near pair $(x,y)$ that is found at the $i$-th recursive level, with
$i\geq 0$, the pair is emitted with probability $1/(p(x,y)L)^i$; otherwise, we ignore it.
For the cache-aware algorithm, the idea is the same but a near pair is emitted with probability $1/(p(x,y)L')$ with $L' = (N/M)^{\rho}$. 

\begin{theorem}
The above approaches guarantee that each near pair is emitted with constant probability in both \ASimJoin\ and \SimJoin.
\end{theorem}
\begin{proof}
The claim easily follows for the cache-aware algorithm: indeed the two points of a near pair $(x,y)$ have the same hash value in $p(x,y)L$ (in expectation) of the $L'=(N/M)^\rho$ repetitions of Step~\ref{step:Arepeat}. 
Therefore, by emitting the  pair with probability $1/(p(x,y)L)$ we get the claim.

We now focus on the cache-oblivious algorithm, where the claim requires a more articulated proof. 
Given a near pair $(x,y)$, let $G_i$ and $H_i$ be random variables denoting respectively the number of subproblems at level $i$ containing the pair $(x,y)$, and the number of subproblems at level $i$ where $(x,y)$ is not found by the cache-oblivious nested loop join algorithm in Theorem~\ref{thm:co-nested-loop}.
Let also $K_i$ be a random variable denoting the actual number of times the pair $(x,y)$ is emitted at level $i$. 
We have followings properties: 
\begin{enumerate}
	\item $\E{K_i| G_i,H_i}=(G_i-H_i)/(p(x,y)L)^i$ since a near pair is emitted with probability $1/(p(x,y)L)^i$ only in those subproblems where the pair is found by the join algorithm.
	\item $\E{G_i}=(p(x,y)L)^i$ since a near pair is in the same bucket with probability $p(x,y)^i$ (it follows from the previous analysis based on standard branching).
	\item $G_0=1$ since each pair exists at the beginning of the algorithm.
	\item $H_\Psi=0$ since each pair surviving up to the last recursive level is found by the nested loop join algorithm.
\end{enumerate}

We are interested in upper bounding $\E{\sum_{i=0}^{\Psi} K_i}$ by induction that
$$\E{\sum_{i=0}^{l} K_i}=1-\frac{\E{H_l}}{(p(x,y)L)^l} \enspace ,$$
for any $0\leq l\leq \Psi$.
For $l=0$ (i.e., the first call to \SimJoin) and note that $\E{G_0}=G_0=1$, the equality is verified since
$$
\E{K_0}=\E{\E{K_0 | G_0, H_0}} = \E{G_0-H_0} = 1-\E{H_0} \enspace.
$$

We now consider a generic level $l>0$. 
Since a pair  propagates in a lower recursive level with probability $p(x,y)$, we have \[\E{G_l} = \E{\E{G_l | H_{l-1}}} = p(x,y)L\E{H_{l-1}}\text{.}\] 
Thus,
\begin{align*}
\E{K_l} = \E{\E{K_l|G_{l},H_{l}}} &=
\E{\frac{G_{l}-H_{l}}{(p(x,y)L)^l}} \\ &=
\frac{\E{H_{\ell-1}}}{(p(x,y)L)^{l-1}} -
\frac{\E{H_{\ell}}}{{(p(x,y)L)^l}} \enspace \enspace .
\end{align*}
By exploiting the inductive hypothesis, we get
$$
\E{\sum_{i=0}^{l} K_i}=
\E{K_l}+\E{\sum_{i=0}^{l-1} K_i}=
1-\frac{\E{H_l}}{(p(x,y)L)^l}.$$
Since $H_\Psi=0$, we have $\E{\sum_{i=0}^{\Psi} K_i}=1$ and the claim follows. \qed
\end{proof}

We observe that the proposed approach is equivalent to use an LSH where $p(x,y)=p_1$ for each near pair.
Finally, we  remark that this approach does not avoid replica of the same near pair when the algorithm is repeated for increasing the collision probability of near pairs. 
Thus, the probability of emitting a pair is at least $\BOM{1/\sqrt{\Psi}}$ as shown in the second part of Section~\ref{sec:correctness} and $\BO{\log^{3/2} N}$ repetitions of \SimJoin\ suffices to find all pairs with high probability (however, the expected number of replica of a given near pair becomes $\BO{\log^{3/2} N}$, even with the proposed approach).

\section{Discussion}\label{sec:exp}

We will argue informally that our I/O complexity of Theorem~\ref{thm:main} is
close to the optimal. For simple arguments, we split the I/O complexity of our
algorithms in two parts:
\begin{align*}
T_1 &= \left(\frac{N}{M}\right)^\rho \left( \frac{N}{B} + \frac{|R \underset{\leq r}{\bowtie} S|}{M B} \right), \\
T_2 &= \frac{|R \underset{\leq cr}{\bowtie} S|}{M B}.
\end{align*}

We now argue that $T_1$ I/Os are necessary. 
First, notice that we need $\BO{N/B}$ I/Os per hash function for transferring data between memories, computing and writing hash values to disk to find collisions. 
Second, since each I/O brings at most $B$ points in order to compute the distance with $M$ points residing in the internal memory, we need $N/B$ I/Os to examine $MN$ pairs. 
This means that when the collision probability of far pairs $p_2 \leq M/N$ and the number of collisions of far pairs is at most $MN$ in expectation, we only need $\BO{N/B}$ I/Os to detect such far pairs. 
Now we consider the case where there are $\BOM{N^2}$ pairs at distance $cr$. 
Due to the monotonicity of LSH family, the collision probability for each such pair must be $\BO{M/N}$ to ensure that $\BO{N/B}$ I/Os suffices to examine such pairs.
In turn, this means that the collision probability for near pairs within distance $r$ must be at most $\BO{(M/N)^\rho}$.
So we need $\BOM{(N/M)^\rho}$ repetitions (different hash functions) to expect at least one collision for any near pair.

Then, a worst-case data set can be given so that we might need to examine, for each of the $\BOM{(N/M)^\rho}$ hash functions, a constant fraction the pairs in $R \bowtie_{\leq r} S$ whose collision probability is constant.
For example, this can happen if $R$ and $S$ include two clusters of very near points.
One could speculate that some pairs could be marked as ``finished'' during computation such that we do not have to compute their distances again.
However, it seems hard to make this idea work for an arbitrary distance measure where there may be very little structure for the output set, hence the $\BO{ |R \bowtie_{\leq r} S|/(M B)}$ additional I/Os per repetition is needed. 



In order to argue that the term $T_2$ is needed, we consider the case where all pairs in $R \bowtie_{\leq cr} S$ have distance $r+\varepsilon$ for a value $\varepsilon$ small enough to make the collision probability of pair at distance $r+\varepsilon$ indistinguishable from the collision probability of pair at distance~$r$.
Then every pair in $R \bowtie_{\leq cr} S$ must be brought into the internal memory to ensure the correct result, which requires $T_2$ I/Os. 
This holds for \emph{any} algorithm enumerating or listing the near pairs. 
Therefore, there does not exist an algorithm that beats the quadratic dependency on $N$ for such worst-case input sets, unless the distribution of the input is known beforehand. 
However, when $|R \bowtie_{\leq cr} S|$ is subquadratic regarding $N$, a potential approach to achieve subquadratic dependency in expectation for similarity join problem is filtering invalid pairs based on their distances --- currently LSH-based method is the only way to do this.

Note that when $M=N$ our I/O cost is $\BO{N/B}$ as we would expect, since just reading the input is optimal. 
At the other extreme, when $B=M=1$ our bound matches the time complexity of internal memory techniques. 
When $|R \bowtie_{\leq cr} S|$ are bounded by $MN$ then our algorithm achieves subquadratic dependency on $N/M$. 
Such an assumption is realistic in some real-world datasets as shown in the experimental evaluation section.


\smallskip

\ms{To complement the above discussion we will evaluate our complexity by computing explicit constants and then evaluating the total number of I/Os spent by analyzing real datasets. Performing these ``simulated experiments'' has the advantage over real experiments that we are not impacted by any properties of a physical machine.}
We again split the I/O complexity of our algorithms in two parts:
\begin{align*}
T_1 &= \left(\frac{N}{M}\right)^\rho \left( \frac{N}{B}+\frac{|R \underset{\leq r}{\bowtie} S|}{M B}\right)\\
T_2 &= \frac{|R \underset{\leq cr}{\bowtie} S|}{M B}
\end{align*}
and carry out experiments to demonstrate that the first term $T_1$ often dominates the second term $T_2$ in real datasets. 
In particular, we depict the cumulative distribution function (cdf) in log-log scale of all pairwise distances (i.e., $\ell_1$, $\ell_2$) and all pairwise similarities (i.e., Jaccard and cosine) on two commonly used datasets: Enron Email\footnote{https://archive.ics.uci.edu/ml/datasets/Bag+of+Words} and MNIST\footnote{http://yann.lecun.com/exdb/mnist/}, as shown in Figure~\ref{fig:SSJ}.
Since the Enron data set does not have a fixed data size per point, we consider a version of the data set where the dimension has been reduced such that each vector has a fixed size.

\begin{figure*}[t]
\centering
\includegraphics[width=1.0\textwidth]{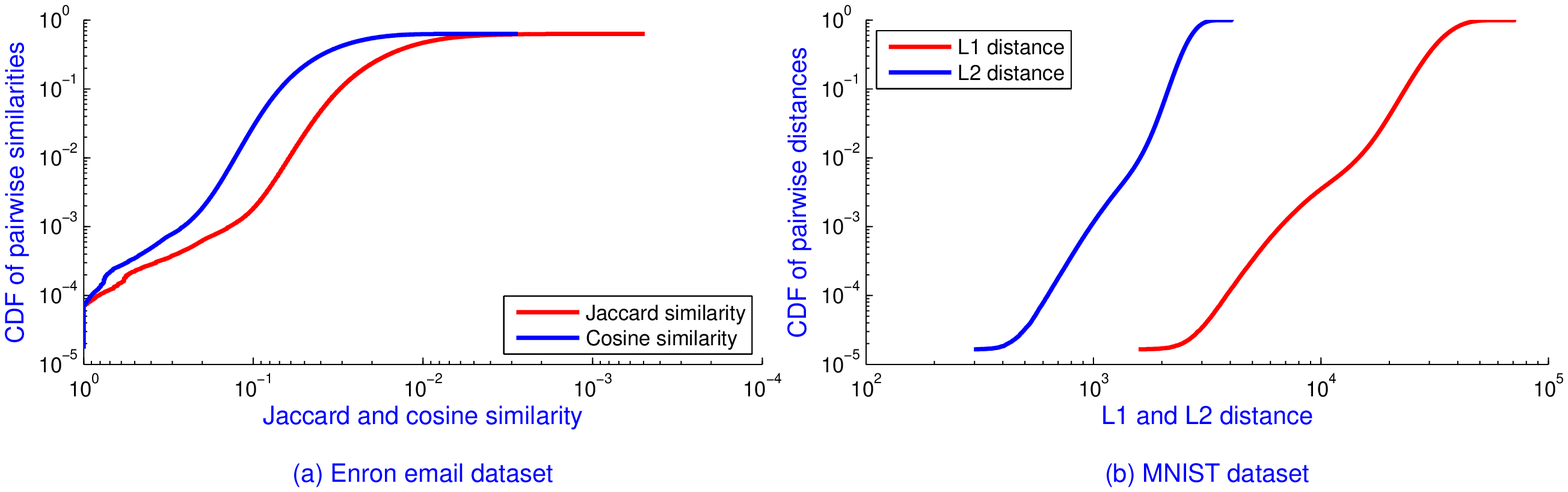}
\caption{The cumulative distributions of pairwise similarities and pairwise distances on samples of 10,000 points from Enron Email and MNIST datasets. We note that values decrease on the x-axis of Figure 1.a, while they increase in Figure 1.b.}
\label{fig:SSJ}
\end{figure*}

Figure~\ref{fig:SSJ}.a shows an inverse \textit{polynomial} relationship with a small exponent $m$ between similarity threshold $s$ and the number of pairwise similarities greater than $s$.  
The degree of the polynomial is particularly low when $s > 0.5$. 
This setting $s > 0.5$ is commonly used in many applications for both Jaccard and cosine similarities~\cite{Arasu_VLDB06,Bayardo_WWW07,Xiao_WWW08}. Similarly, Figure~\ref{fig:SSJ}.b also shows a \textit{monomial} relationship between the distance threshold $r$ and the number of pairwise distances smaller than $r$. 
In turn, this means that the number of $c$-near pairs $|R \bowtie_{\leq cr} S|$ is not much greater than $c^m|R \bowtie_{\leq r} S|$. 
In other words, the second term $T_2$ is often much smaller than the first term $T_1$.

Finally, for the same data sets and metrics, we simulated the cache-aware algorithm with explicit constants and examined the I/Os cost to compare with a standard nested loop method (Section~\ref{sec:generic}) and a lower bound on the standard LSH method (Section~\ref{sec:sorting}).
We set the cache size $M = N/1000$, which is reasonable for judging a number of cache misses since the size ratio between CPU caches and RAM is in that order of magnitude.
In general such setting allows us to investigate what happens when the data size is much larger than fast memory.
For simplicity we use $B=1$ since all methods contain a multiplicative factor $1/B$ on the I/O complexity. 
The values of $\rho$ were computed using good LSH families for the specific metric and parameters $r$ and $c$.
These parameters are picked according to Figure~\ref{fig:SSJ} such that the number of $c$-near pairs are only an order of magnitude larger than the number of near pairs.

The I/O complexity used for nested loop join is $2N + N^2/MB$ (here we assume both sets have size $N$) and the complexity for the standard LSH approach is \emph{lower bounded} by $\sort{N^{1 + \rho}}$.
This complexity is a lower bound on the standard sorting based approaches as it lacks the additional cost that depends on how LSH distributes the points.
Since $M = N / 1000$ we can bound the $\log$-factor of the sorting complexity and use $\sort{N} \leq 8N$ since $2N$ points read and written twice. 
The I/O complexity of our approach is stated in Theorem~\ref{thm:aware}. 
The computed I/O-values in Figure~2 show that the complexity of our algorithm is lower than that of all instances examined.
Nested loop suffers from quadratic dependency on $N$, while the standard LSH bounds lack the dependency on $M$. 
Overall the I/O cost indicates that our cache-aware algorithm is practical on the examined data sets.

\begin{figure*}[t]
\begin{center}
\resizebox{\columnwidth}{!}{%
\begin{tabular}{|c|c|c|c|c|c|c|c|c|c|}
\hline
Data set & Metric & $r$ & $cr$ & $\rho$ & $\frac{|R \underset{\leq r}{\bowtie} S|}{N}$ & $\frac{|R \underset{\leq cr}{\bowtie} S|}{N}$ & Standard LSH & Nested loop & ASimJoin\\
\hline
\hline
Enron & Jaccard & $0.5$ & $0.1$ & $0.30$ & $1.8\cdot 10^3$ & $16\cdot 10^3$ & $> 7.5 \cdot 10^9$ I/Os & $8 \cdot 10^9$ I/Os & $3.2 \cdot 10^9$ I/Os\\
\hline
Enron & Cosine & $0.7$ & $0.2$ & $0.51$ & $1.6\cdot 10^3$ & $16\cdot 10^3$ & $> 212 \cdot 10^9$ I/Os & $8 \cdot 10^9$ I/Os & $6.6 \cdot 10^9$ I/Os\\
\hline
MNIST & L1 & $3000$ & $6000$ & $0.50$ & $1.8$ & $42$ & $> 29 \cdot 10^6$ I/Os & $60 \cdot 10^6$ I/Os & $12 \cdot 10^6$ I/Os\\
\hline
\end{tabular}}
\label{fig:ASSJ}
\caption{A comparison of I/O cost for similarity joins on the standard LSH, nested loop and \ASimJoin\ algorithms.}
\end{center}
\end{figure*}


\section{Conclusion}\label{sec:concl}
In this paper we examine the problem of computing the similarity join of two relations in an external memory setting. 
Our new cache-aware algorithm of Section~\ref{sec:simple} and cache-oblivious algorithm of Section~\ref{sec:cacheobl} improve upon current state of the art by around a factor of $(M/B)^\rho$ I/Os unless the number of $c$-near pairs is huge (more than $NM$).
We believe this is the first cache-oblivious algorithm for similarity join, and more importantly the first subquadratic algorithm whose I/O performance improves significantly when the size of internal memory grows.

It would be interesting to investigate if our cache-oblivious approach is also practical --- this might require adjusting parameters such as $L$.
Our I/O bound is probably not easy to improve significantly, but interesting open problems are to remove the error probability of the algorithm and to improve the implicit dependence on dimension in $B$ and $M$. Note that our work assumes for simplicity that the unit of $M$ and $B$ is number of points, but in general we may get tighter bounds by taking into account the gap between the space required to store a point and the space for hash values.
Also, the result in this paper is made with general spaces in mind and it is an interesting direction to examine if the dependence on dimension could be made explicit and improved in specific spaces.

\bibliographystyle{plain}
\bibliography{simjoin}

\end{document}